%% file: main.tex
\documentclass[submission,copyright,creativecommons]{eptcs}
\providecommand{\event}{TARK 2025} 
\input{macros.tex}

\usepackage{boondox-cal}
\usepackage[bb=boondox]{mathalfa}
\usepackage{xfrac}
\usepackage{enumitem}
\usepackage[all]{xy}
\usepackage{stmaryrd}
\usepackage{iftex}
\usepackage{mathrsfs}
\usepackage{amsmath}
\usepackage{amsthm}
\usepackage{amssymb}
\usepackage{mathtools}
\usepackage{longtable}
\usepackage{xcolor}
\usepackage{tikz}

\usepackage{bm} 
\usetikzlibrary{positioning,arrows,arrows.meta}
\ifpdf
  \usepackage{underscore}         
  \usepackage[T1]{fontenc}        
\else
  \usepackage{breakurl}           
\fi

\theoremstyle{definition}
\newtheorem{definition}{Definition}
\theoremstyle{remark}
\newtheorem{remark}{Remark}
\newtheorem{example}{Example}
\newtheorem{convention}{Convention}
\usepackage{thm-restate}
\tikzset{
   	state/.style={
	circle,
	inner sep=1pt, minimum width=8mm,
	draw=black, fill=black!15,
	font=\small
    }
}
\newcommand{\tot}{\leftrightarrow} 
\title{Complexity of Łukasiewicz Modal Probabilistic Logics}
\author{Daniil Kozhemiachenko
\institute{Aix Marseille Univ, CNRS, LIS, Marseille, France}
\email{daniil.kozhemiachenko@lis-lab.fr}
\and
Igor Sedlár\thanks{Work of the second author was supported by the grant 22-16111S of the Czech Science Foundation. We would like to thank the reviewers for their comments, which helped us improve the paper. 
}
\institute{The Czech Academy of Sciences, Institute of Computer Science, Prague, Czech Republic}
\email{sedlar@cs.cas.cz}
}
\def\titlerunning{Complexity of Łukasiewicz Modal Probabilistic Logics}
\def\authorrunning{D.\ Kozhemiachenko \& I. Sedlár}
\begin{document}
\maketitle
\begin{abstract}
Modal probabilistic logics provide a framework for reasoning about probability in modal contexts, involving notions such as knowledge, belief, time, and action. In this paper, we study a particular family of these logics, extending the modal {\L}ukasiewicz many-valued logic. These logics are shown to be capable of expressing nuanced probabilistic concepts, including upper and lower probabilities. Our main contribution is a $\pspace$-completeness result for two variants of the local consequence problem, providing a precise computational characterisation.
\end{abstract}

\section{Introduction\label{sec:introduction}}
Probabilistic logics have been studied and applied in various fields for decades~\cite{Halpern2017,OgnjanovicEtAl2016}. In particular, they provide formal languages that express information about probabilistic models in a compact yet unambiguous way, and they can be used to facilitate and verify reasoning about these models. The \emph{many-valued approach} to probabilistic logic goes back to the work of Hájek et al.~\cite{Hajek1998,HajekEtAl1995,HajekEtAl2000,HajekHarmancova1995}, where `$\alpha$ is probable' is taken as an \emph{imprecise statement} whose truth value (or truth \emph{degree}) is identified with the probability of $\alpha$. This approach contrasts with the better-known logic of Fagin et al.~\cite{FaginHalpernMegiddo1990}, which uses formulas involving linear inequalities over probabilistic terms to represent \emph{precise statements} about probability -- statements which are either true or false. Many-valued probabilistic logics are tailored for expressing \emph{imprecise} information about probabilistic models in terms of properties that models do not simply ``have'' or ``not have'', but which models have \emph{to some degree}. A particular example are imprecise comparisons, indicating the degree to which two probabilities differ. Another advantage of the many-valued approach is its simpler syntax and axiomatisation. Nevertheless, as shown in \cite{BaldiCintulaNoguera2020}, the many-valued approach and the linear approach are linked by entailment-preserving translations.

In computer science, artificial intelligence and economics, \emph{modal} probabilistic logics are of particular importance. These logics formalise reasoning about probability in the presence of modal notions such as knowledge, belief, time and action \cite{Aumann1999a,BacchusEtAl1999,DoderPerovic2020,FaginHalpern1994}. In \cite{MajerSedlar2025}, a~many-valued framework for modal probabilistic logic was introduced by combining modal {\L}ukasiewicz logic \cite{HansoulTeheux2013} with the propositional probabilistic logics of Hájek et al. Decidability of the basic many-valued modal probabilistic logic and some of its extensions was established by reduction to basic non-probabilistic modal {\L}ukasiewicz logic, which is known to be decidable \cite{Vidal2021}. However, the reduction is not polynomial and so it does not by itself yield complexity results for many-valued modal probabilistic logics. This paper makes two contributions to this line of research. First, using the tableaux method, we show that the problem of deciding local consequence over \emph{finitely branching} frames is $\pspace$-complete even for a~rather expressive extension of the language used in \cite{MajerSedlar2025}. We also show that the problem of deciding local consequence over arbitrary frames is $\pspace$-complete if we restrict ourselves to a~specific fragment of the full language. Secondly, we argue that our framework is sufficiently expressive to formalise a~variety of practically relevant concepts and scenarios, including upper and lower probabilities and their imprecise comparisons. These two contributions together establish many-valued modal logics as a~viable framework for formalising the interplay between probability and modality. 

The paper is structured as follows. Section \ref{sec:frames} introduces our basic semantic structures, probabilistic frames, and provides several examples of these structures. It also discusses their relationship with similar structures proposed in the literature. Section \ref{sec:logics} discusses a~many-valued modal logic for reasoning about probabilistic frames. This logic is based on a~modal extension of $\Luk\bm{\Pi}\onehalf$ \cite{EstevaEtAl2001}, which is an expressive combination of {\L}ukasiewicz and product fuzzy logics \cite{Hajek1998}. Section~\ref{sec:express} shows how our logic can be used to represent and reason about a~variety of practically relevant concepts and scenarios. Section \ref{sec:results} establishes our main technical results on the $\pspace$-completeness of the local consequence problem (i) over finitely branching frames using the full language and (ii) over arbitrary frames using a~specific fragment of the full language. Section~\ref{sec:conclusion} concludes the paper and lists some tasks left for the future.

\section{Probabilistic frames}\label{sec:frames}

Recall that a~\emph{finitely additive probability measure} on a~Boolean algebra $X$ is a~function $P: X \to [0,1]$ such that (i) $P(1_X) = 1$ and (ii) $P(x \lor y) = P(x) + P(y)$ for all $x,y$ such that $x \land y = 0_X$. A probability measure on a finite Boolean algebra is uniquely determined by its values on the atoms of the algebra. 

\begin{definition}[Probabilistic frames]\label{def:probabilisticframes}
Let $\Amsf$ be an at most countable set. A~\emph{probabilistic $\Amsf$-frame} is a~tuple $\Ffrak=\langle W,R,S,\mu\rangle$ such that $\langle W,R\rangle$ is a~Kripke $\Amsf$-frame comprising a~non-empty set $W$ and a~function $R$ from $\Amsf$ to binary relations on $W$, $S$ is a~Boolean subalgebra of $2^W$ containing $W$, and $\mu=\langle\mu_w\rangle_{w\in W}$ with each $\mu_w:S\rightarrow[0,1]$ being a~finitely additive probability measure on~$S$.
\end{definition}

The set of modal indices $\Amsf$ can be used to represent various kinds of objects, for instance, agents or actions. As usual in modal logic, Kripke $\Amsf$-frames represent a~collection of possible states of affairs, or \emph{possible worlds}, connected with \emph{accessibility relations} $R_a$ for $a \in \Amsf$. These relations can represent, for example, the qualitative uncertainty of agents ($wR_au$ means that the information of agent $a$ in world $w$ is not sufficient to exclude $u$ as a~possibility) or the effects of actions ($wR_au$ means that world $u$ is a~possible outcome of performing action $a$ in world $w$). The collection $S$ of subsets of $W$ represents \emph{events} (or \emph{propositions}). Intuitively, events correspond to features of possible worlds to which probabilities will be assigned. This assignment is carried out using the collection of functions $\mu$: for every world $w \in W$, $\mu_w(E)$ expresses the probability of event $E \in S$ \emph{in} world $w$ --- the crucial feature of \emph{modal} probabilistic frames is that the probabilities of events can vary from world to world, much as the truth values of propositions can vary from world to world in standard modal semantics. 

\begin{example}\label{exam:robot1}
Suppose a~robot is sent to retrieve an item from a~warehouse. If the warehouse light is on (i.e.~$L$ is true), then the probability of retrieving the correct item (i.e.~$I$ being true) is $0.8$. On the other hand, if $L$ is false, then the probability of $I$ being true is $0.2$. There are two possible start states, which represent the situation before the robot operates and are distinguished by whether the light is on or off. There are also four possible end states, as shown in Figure~\ref{fig:robot-example}. The start states, $s_L$ and $s_{\neg L}$, are shown on the bottom line, and the end states $e_{L \land I}$, $e_{L \land \neg I}$, $e_{\neg L \land I}$ and $e_{\neg L \land \neg I}$ are shown on the top line. The arrows labelled with $b$ represent the possible outcomes of the robot's actions (it is assumed that the robot cannot turn the light switch). In the state $s_{L}$, the probability of the event $ E_I = \{ e_{L \land I}, e_{\neg L \land I}\}$, indicated using the thick circle, is $0.8$ and the probability of the event $E_{\neg I} = \{ e_{L \land \neg I}, e_{\neg L \land \neg I} \}$ is $0.2$. In the state $s_{\neg L}$, the probabilities are reversed: $E_{I}$ has probability $0.2$ and $E_{\neg I}$ has probability $0.8$.\footnote{We assume for the sake of simplicity that, in each end state $e$, the only event that has non-zero probability is $\{ e \}$ and that the probability of this event is $1$.} The arrows labelled by $a$ represent the action of switching the warehouse light switch before the robot operates. This action clearly changes the probability of $E_{I}$~--~the probability of $E_I$ increases when the light is switched on and decreases when it is switched off. 
\end{example}

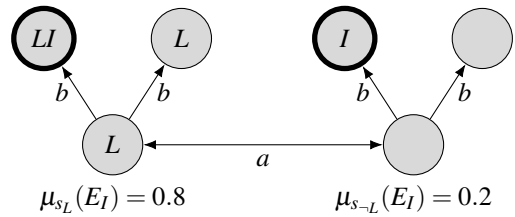
\begin{figure}
    \centering
    \begin{tikzpicture}[node distance=2cm]
        \node[state,label={below:{\small $\mu_{s_L}(E_{I}) = 0.8$ }}] (inL1) at (0,0) {$L$};
        \node[state,right of=inL1,xshift=2cm,label={below:{\small $\mu_{s_{\neg L}}(E_{I}) = 0.2$ }}] (inL0) { };
        \node[state,line width=2pt,above left of=inL1,xshift=5mm] (L1outI1) {$LI$};
        \node[state,above right of=inL1,xshift=-5mm] (L1outI0) {$L$};
        \node[state,line width=2pt,above left of=inL0,xshift=5mm] (L0outI1) {$I$};
        \node[state,above right of=inL0,xshift=-5mm] (L0outI0) { };
        \path[-Latex]
        (inL1) edge node[left] {\small $b$} (L1outI1)
        (inL1) edge node[right] {\small $b$} (L1outI0);
        \path[-Latex]
        (inL0) edge node[left] {\small $b$} (L0outI1)
        (inL0) edge node[right] {\small $b$} (L0outI0);
        \path[Latex-Latex]
        (inL1) edge node[below] {\small $a$} (inL0);
    \end{tikzpicture}
    \caption{The robot in a~warehouse example.}
    \label{fig:robot-example}
\end{figure}

\begin{example}\label{exam:robot2}
Take the scenario from Example~\ref{exam:robot1}, but now suppose we have an agent $a$ who lacks information about the warehouse lights. In standard probabilistic fashion, $a$'s uncertainty could be modelled \emph{quantitatively} by a~subjective probability distribution on $\{ s_L, s_{\neg L} \}$. However, such a~distribution expresses more information than we seem to have. We know that $s_L$ and $s_{\neg L}$ are possibilities for the agent, but we don't know how likely the agent thinks these possibilities are. Without this additional information, we can use the resources of epistemic modal logic and model the agent's uncertainty \emph{qualitatively} by a~binary equivalence relation. This is again depicted in Figure~\ref{fig:robot-example} if the arrow labelled by $a$ is seen as representing agent $a$'s qualitative uncertainty. The idea of modelling qualitative uncertainty by a~\emph{set} of probability distributions is an old one \cite[Ch.~2.3]{Halpern2017}.
\end{example}

\begin{example}\label{exam:MarkovChain}
Well-known kinds of probabilistic relational structures also provide examples of probabilistic frames. Example~\ref{exam:robot1} suggests the following construction. A \emph{discrete-time Markov chain} is a~triple $\Mmbf=\langle\{ 1, \ldots, n \},\pfrak,\qfrak\rangle$ where $N = \{ 1, \ldots, n \}$ represents a~set of states, $\pfrak:N\times N\to [0,1]$ is a~transition probability function such that $\sum_{j = 1}^{n}\pfrak(i,j) = 1$ for all $i \in  N$, and $\qfrak: N \to [0,1]$ is a~start probability function such that $\sum_{i = 1}^{n}\qfrak(i) = 1$. For each $m \in \Nmbb$, a~Markov chain $\Mmbf$ gives rise to a~probability measure on the set of all paths over $N$ of length $m$: $P(i_0,\ldots,i_m) = q(i_0) \cdot \prod_{j = 1}^{m}\pfrak(i_{j-1},i_j)$. Every $\Mmbf=\langle N=\{1,\ldots,n\},\pfrak,\qfrak\rangle$ can be seen as a~probabilistic $N$-frame $\Ffrak_\Mmbf=\langle N_0,R,2^S,\langle\mu_i\rangle_{i \in N_0}\rangle$ where $N_0 = N \cup \{ 0 \}$, $R_k(i,j)$ iff $j = k$ for $j,k \in N$ and $i \in N_0$, and $\mu_i (\{ j \}) =\pfrak(i,j)$ if $i,j \in N$ and $\mu_0(\{ i \})=\qfrak(i)$. 
\end{example}
\begin{example}\label{example:multiagentprobabilisticframes}
The following is a~generalisation of probabilistic frames. A \emph{multi-agent probabilistic $\Amsf$-frame} is a~structure $\langle W, R, S, \mu \rangle$ such that $\langle W, R, S \rangle$ is as before and $\mu = \langle \mu_{w,a} \rangle_{w \in W, a~\in \Gmsf}$ for some $\Gmsf \subseteq \Amsf$ such that $\mu_{w,a} : S \to [0,1]$ is a~finitely additive probability measure. Modal indices in $\Gmsf$ represent agents and $\mu_{w,a}$ represents the subjective probability distribution of agent $a$ in $w$. That is, in multi-agent frames, probability measures $\mu$ have a world index \emph{and} an agent index while accessibility relations $R$ are indexed by all $a \in \Amsf$ as before. Probabilistic frames correspond to the special case with only one fixed agent index $a_0 \in \Amsf$ which can be omitted. Multi-agent probabilistic frames where $W$ is finite, $S = 2^W$, $\Amsf = \Gmsf$ and $R_a$ is an equivalence relation correspond to finite semantic knowledge-belief systems of \cite{Aumann1999a}, structures that extend finite Harsanyi type spaces \cite{HeifetzMongin2001} with S5-style knowledge operators. Intuitively, the relation $R_a$ represents the knowledge state of agent $a$ in a~manner similar to Example~\ref{exam:robot2}. We confine the more general multi-agent framework to this example mainly for reasons of notational simplicity.
\end{example}


\begin{remark}
Probabilistic frames are similar to structures that have appeared in the literature. Aumann's \emph{infinite} semantic knowledge-based systems \cite{Aumann1999a} assume that $S$ is closed under countable unions ($\sigma$-algebra) and that $\mu_a (\cdot, E) : w \mapsto \mu_{w,a}(E)$ is an $S$-measurable mapping from $W$ to $[0,1]$. This condition also appears in the definition of \emph{Markov kernels}, probabilistic transition systems on infinite state spaces. The condition allows to define probabilities of events of the form ``The probability of event $E$ is at least $r$'' for $r \in [0,1] \cap \Qmbb$; see \cite{HeifetzMongin2001,KozenEtAl2013,Zhou2007}. In this paper we focus only on probability assignments to Boolean events, and so we can use simpler structures. Fagin and Halpern~\cite{FaginHalpern1994} use similar $\Amsf$-frames $\langle W, R, \Pmc\rangle$ with $\Pmc(a, w)=\langle W_{a,w}, S_{a,w}, \mu_{a,w}\rangle$ being a~probability space for $W_{a,w}\subseteq W$, $S_{a,w}$ a~Boolean algebra of subsets of $W_{a,w}$ containing $W_{a,w}$, and $\mu_{a,w}$ a~probability measure on $S_{a,w}$. That is, not only the probability measure but also the space of measurable events can vary from world to world.\footnote{We note that Fagin and Halpern's frames do not require that $\mu_a (\cdot, E): w \mapsto \mu_{w,a}(E)$ be measurable; indeed, such a~requirement requires a~``common'' space of measurable events. Instead, they use inner measures to define probabilities of events of the form ``the probability of event $E$ is at least $r$''.} Our structures are simpler.   
\end{remark}

\section{Modal many-valued probabilistic logic}\label{sec:logics}
We have seen that probabilistic frames are a~general and versatile type of semantic structure. In this section, we introduce a~language based on many-valued modal logic for expressing and reasoning about properties of probabilistic frames. We build on the work of Hájek et al.~\cite{Hajek1998,HajekEtAl1995,HajekEtAl2000,HajekHarmancova1995}, who focus on variants of the well-known {\L}ukasiewicz fuzzy logic and its combinations with product fuzzy logic. The propositional fragment of our modal language can be seen as the language of $\Luk\bm{\Pi}\onehalf$ logic, a~combination of {\L}ukasiewicz and product logic and one of the most expressive propositional fuzzy languages.%


Let $\Prop$ be a~countably infinite set of propositional variables. Let $\LCPL$ be the set of Boolean terms generated via the following grammar:
\begin{align*}
\LCPL\ni\alpha&\coloneqq p\in\Prop\mid{\sim}\alpha\mid(\alpha\wedge
\alpha)
\end{align*}
We use Greek letters $\alpha, \beta$ to range over $\LCPL$.

\begin{definition}
The modal probabilistic language $\LProb^{\Amsf}$ is given by 
\begin{align*}
\LProb^{\Amsf}\ni\phi&\coloneqq\Prob(\alpha)\mid\neg\phi\mid(\phi\rightarrow\phi)\mid(\phi\bullet\phi)\mid(\phi\rightarrow_\Pi\phi)\mid\Box_a\phi\mid\onehalf
\end{align*}
for $\alpha \in \LCPL$ and $a \in \Amsf$. We will assume a~fixed $\Amsf$ from now on and omit explicit reference to it in most notation. In what follows, we will call formulas of the form $\Prob(\alpha)$ \emph{probabilistic atoms}.
\end{definition}

\begin{convention}[Notation]\label{conv:notation}
Given $\phi\in\LProb$, we use $\Prop(\phi)$ to denote the set of propositional variables occurring in~$\phi$. Furthermore, $\lmc(\phi)$ denotes the length of~$\phi$, i.e., the number of occurrences of symbols in it. Similarly, for $\Gamma\subseteq\LProb$, we use $\Prob[\Gamma]$ and $\lmc[\Gamma]$ for the set of propositional variables in~$\Gamma$ and its length. Finally, given a~binary relation $R$ on~$W$, we set $R(w)\coloneqq\{w'\mid wRw'\}$.
\end{convention}
\begin{definition}[Modal depth]\label{def:modaldepth}
Let $\phi\in\LProb$. Its \emph{modal depth} ($\dmc(\phi)$) is defined inductively as follows: $\dmc(p)=0$, $\dmc(\neg\phi)=\dmc(\phi)$, $\dmc(\phi\circ\chi)=\max(\dmc(\phi),\dmc(\chi))$ (with $\circ\in\{\rightarrow,\bullet,\rightarrow_\Pi\}$), $\dmc(\Box_a\phi)=\dmc(\phi)+1$. 
Given $\Gamma\subseteq\LProb$, the modal depth of $\Gamma$ is ($\dmc[\Gamma]$) is defined as follows: $\dmc[\Gamma]=\max\{\dmc(\phi)\mid\phi\in\Gamma\}$.
\end{definition}

\begin{definition}[Semantics of $\KLukProdProb$]\label{def:KLukProbmodels}
A~\emph{$\KLukProdProb$-model} is a~pair $\Mfrak=\langle\Ffrak,E\rangle$ with $\Ffrak$ a~\emph{probabilistic} frame and an \emph{event function} or \emph{valuation} $E:\Prop\rightarrow S$ extended to a~homomorphism $\LCPL \rightarrow S$ in an obvious way. The~\emph{probabilistic interpretation induced by $\Mfrak$} is a~function $\Imc_\Mfrak:\LProb\times W\rightarrow[0,1]$ s.t.\ $\Imc_\Mfrak(\onehalf,w)=\frac{1}{2}$ and
\begin{align*}
\Imc_\Mfrak(\Prob(\alpha),w) &= \mu_w(E(\alpha)) &
\Imc_\Mfrak(\phi\!\rightarrow\!\chi,w)&=\min(1,1-\Imc_\Mfrak(\phi,w)+\Imc_\Mfrak(\chi,w))\\
\Imc_\Mfrak(\phi\bullet\chi,w)&=\Imc_\Mfrak(\phi,w)\cdot\Imc_\Mfrak(\chi,w)&\Imc_\Mfrak(\phi\rightarrow_\Pi\chi,w)&=\begin{cases}1&\text{if }\Imc_\Mfrak(\phi,w)\leq\Imc_\Mfrak(\chi,w)\\\frac{\Imc_\Mfrak(\chi,w)}{\Imc_\Mfrak(\phi,w)}&\text{otherwise}\end{cases}\\
\Imc_\Mfrak(\neg\phi,w)&=1-\Imc_\Mfrak(\phi,w) &
\Imc_\Mfrak(\Box_a\phi,w)&=\inf\{\Imc_\Mfrak(\phi,w')\mid w'\in R_a(w)\}
\end{align*}
We say that $\phi\in\LProb$ is \emph{$\KLukProdProb$-valid on a~pointed frame~$\langle\Ffrak,w\rangle$} (denoted $\Ffrak,w\models_\Prob\phi$) if $\Imc_\Mfrak(\phi,w)=1$ for every $E: \Prop\to S$; $\phi$ is \emph{$\KLukProdProb$-valid on~$\Ffrak$} ($\Ffrak\models_\Prob\phi$) if $\Ffrak,w\models_\Prob\phi$ for every $w\in\Ffrak$. Given a~class $\Fmbb$ of probabilistic frames, we say that $\phi$ is \emph{$\KLukProdProb$-valid in~$\Fmbb$} ($\Fmbb\models_\Prob\phi$) if $\Ffrak\models_\Prob\phi$ for every $\Ffrak\in\Fmbb$. Finally, \emph{$\phi$ is $\KLukProdProb$-valid} ($\KLukProdProb\models\phi$) if it is valid on all frames, and $\Gamma\subseteq\LProb$ \emph{entails} $\chi$ ($\Gamma\models_\Prob\chi$) if $\Imc_\Mfrak(\chi,w)=1$ in every $\KLukProdProb$-model $\Mfrak$ and $w\in\Mfrak$ s.t.\ $\Imc_\Mfrak(\phi,w)=1$ for all $\phi\in\Gamma$.
\end{definition}

Other Łukasiewicz connectives can be defined in a~standard manner:
\begin{align*}
\top&\coloneqq\Prob(p)\rightarrow\Prob(p)&\triangle\phi&\coloneqq\neg\phi\rightarrow_\Pi\bot&\phi\oplus\chi&\coloneqq\neg\phi\rightarrow\chi&\phi\odot\chi&\coloneqq\neg(\neg\phi\oplus\neg\chi)\\
\bot&\coloneqq\neg\top&
\Diamond_a \phi &\coloneqq \neg \Box_a \neg \phi &\phi\vee\chi&\coloneqq(\phi\rightarrow\chi)\rightarrow\chi&\phi\wedge\chi&\coloneqq\neg(\neg\phi\vee\neg\chi)\\
&&&&\phi\leftrightarrow\chi&\coloneqq(\phi\rightarrow\chi)\odot(\chi\rightarrow\phi)
\end{align*}
Using Definition~\ref{def:KLukProbmodels}, one can obtain their semantics:
\begin{align*}
\Imc_\Mfrak(\top,w)&=1&\Imc_\Mfrak(\phi\oplus\chi,w)&=\min(1,\Imc_\Mfrak(\phi,w)+\Imc_\Mfrak(\chi,w))\\
\Imc_\Mfrak(\bot,w)&=0&\Imc_\Mfrak(\phi\odot\chi,w)&=\max(0,\Imc_\Mfrak(\phi,w)+\Imc_\Mfrak(\chi,w)-1)\\
\Imc_\Mfrak(\triangle\phi,w)&=\begin{cases}1&\text{if }\Imc_\Mfrak(\phi,w)=1\\0&\text{otherwise}\end{cases}&
\begin{matrix}
\Imc_\Mfrak(\phi\vee\chi,w)\\
\Imc_\Mfrak(\phi\wedge\chi,w)
\end{matrix}
&
\begin{matrix}
~=\max(\Imc_\Mfrak(\phi,w),\Imc_\Mfrak(\chi,w))\\
~=\min(\Imc_\Mfrak(\phi,w),\Imc_\Mfrak(\chi,w))
\end{matrix}\\
\Imc_\Mfrak(\phi\leftrightarrow\chi,w)&=1-|\Imc_\Mfrak(\phi,w)-\Imc_\Mfrak(\chi,w)|&\Imc_\Mfrak(\Diamond_a \phi, w) & = \sup \{ \Imc_\Mfrak(\phi, w') \mid w' \in R_a (\phi) \}
\end{align*}

We note that the semantics of $\KLukProdProb$ can be equivalently reformulated if we redefine frames by allowing $S$ to be an arbitrary Boolean algebra independent from $W$. We call these frames and their respective models \emph{sample-independent} (SI-frames and SI-models).
\begin{definition}\label{def:sampleindependentmodels}~
\begin{itemize}
\item An \emph{SI probabilistic $\Amsf$-frame over $X$} (SI-frame) is a~tuple $\Ffrak=\langle W,R,X,\mu\rangle$ with $\langle W,R\rangle$ being a~Kripke $\Amsf$-frame, $X\neq\varnothing$ a~Boolean algebra, and $\mu=\langle\mu_w\rangle_{w\in W}$ being a~tuple of finitely additive probability measures on~$X$.
\item An \emph{SI-$\KLukProdProb$-model} is a~tuple $\Mfrak=\langle\Ffrak,E\rangle$ with $\Ffrak$ being a~sample-independent probabilistic frame and $E:\Prop\rightarrow X$. The notions of probabilistic interpretations induced by~$\Mfrak$, entailment, and validity are the same as in Definition~\ref{def:KLukProbmodels}.
\end{itemize}
\end{definition}

\begin{restatable}{theorem}{semanticsequivalence}\label{theorem:semanticsequivalence}
Let $\phi\!\in\!\LProb$. Then $\phi$ is~$\KLukProdProb$-valid iff $\phi$ is~$\KLukProdProb$-valid on SI-models.
\end{restatable}
\begin{proof}
Clearly, if $\phi$ is~$\KLukProdProb$-valid on SI-models, then it is valid on $\KLukProdProb$-models. For the converse, let $\phi$ be \emph{not} $\KLukProdProb$-valid on SI-models. Thus, $\Imc_\Mfrak(\phi,w)=x<1$ for some SI-model $\Mfrak=\langle W,R,X,\mu,E\rangle$. If $|X|\leq|2^W|$, it is clear, that $\Mfrak$ is also a~$\KLukProdProb$-model. Otherwise let $|X|=2^{|W|+n}$. We define a~$\KLukProdProb$ model $\Mfrak'=\langle W',R',X,\mu,E\rangle$ as follows. Let $W'=W\cup\{w_1,\ldots,w_n\}$ for new states $w_1$, \ldots, $w_n$; let further, $wR'_aw'$ iff $\{w,w'\}\subseteq W$ and $wR_aw'$. Now, using that the measures and valuations of propositional variables in~$\Mfrak$ and $\Mfrak'$ coincide, one can show by induction on $\phi$ that $\Imc_\Mfrak(\phi,w)=\Imc_{\Mfrak'}(\phi,w)$ for every $w\in W$. The result follows.
\end{proof}

An important property of SI-models is that given~$\Mfrak$ and $\phi\in\LProb$, one can assume that $X=2^{2^\Prop(\phi)}$. Intuitively, this means that to verify whether an $\LProb$-formula is valid, it suffices to check probability distributions over Boolean terms that can be composed of its variables.
\begin{restatable}{proposition}{simplification}\label{prop:simplification}
Let $\Mfrak=\langle W,R,X,\mu,E\rangle$ be an SI-model. Define $\Mfrak^\Cmsf_\phi=\langle W,R,2^{2^{\Prop(\phi)}},\mu^\Cmsf,E^\Cmsf\rangle$ s.t.
\begin{align*}
\forall p\in\Prop:E^\Cmsf(p)=\{Y\mid p\in Y,~Y\subseteq\Prop(\phi)\}\\
\forall w\in W,Y\subseteq\Prop(\phi):\mu^\Cmsf_w(\{Y\})=\mu_w\bigg(E\bigg(\bigwedge\limits_{p\in Y}\!\!p\wedge\bigwedge\limits_{q\notin Y}\!\!\neg q\bigg)\bigg)
\end{align*}
Then $\Imc_\Mfrak(\phi,w)=\Imc_{\Mfrak^\Cmsf_\phi}(\phi,w)$ for every $w\in W$ and $\phi\in\LProb$.
\end{restatable}
\begin{proof}
Induction on $\phi$. For the basis case, we let w.l.o.g.\ $\phi=\Prob(\alpha)$ for some $\alpha\!=\!\bigvee\limits^n_{i=1}\tau_i$ in the DNF each part of which contains all variables of $\phi$. Now, $\Imc_\Mfrak(\Prob(\alpha),w)\!=\!x$ means that $\mu_w\left(E\left(\bigvee\limits^{n}_{i=1}\tau_i\right)\right)\!=\!x$, i.e., $\sum\limits^n_{i=1}\mu_w(E(\tau_i))\!=\!x$. From here, using the fact that $Y\vDash\tau$ iff $Y=\{p\mid p\in\tau\}$, we obtain that $\mu^\Cmsf_w(E^\Cmsf(\tau_i))=\mu_w((\tau_i))$ for every $i\in\{1,\ldots,n\}$. Thus, $\sum\limits^n_{i=1}\mu^\Cmsf_w(E^\Cmsf(\tau_i))=\sum\limits^n_{i=1}\mu_w(E(\tau_i))=x$, as required.

The cases of propositional connectives can be shown by a~straightforward application of the induction hypothesis. The case of $\phi=\Box\chi$ follows easily since $\langle W,R\rangle$-reducts in $\Mfrak$ and $\Mfrak^\Cmsf_\phi$ coincide. 
\end{proof}

In what follows, we will use $\Lmodal$ to denote the fragment of $\LProb$ s.t.\ only the probabilistic atoms of the form $\Prob(p)$ with $p\in\Prop$ are allowed. We will also use $\KLukProd$ to denote the $\Lmodal$-fragment of $\KLukProdProb$ and $\models$ for its entailment (satisfaction) relation. As every two distinct probabilistic atoms concern two \emph{independent events}, it is clear that $\KLukProd$-models can be built on Kripke frames, \emph{ignoring measures}, i.e., treating each $\Prob(p)$ as a~\emph{propositional variable}. Namely, given a~\emph{Kripke frame} $\Ffrak=\langle W,R\rangle$, and a~map $v:\Prop\times W\rightarrow[0,1]$, a~\emph{$\KLukProd$-model} is a~pair $\Mfrak=\langle\Ffrak,v\rangle$. A~\emph{$\KLukProd$-interpretation $\Imc_\Mfrak$ induced by~$\Mfrak$} is defined as for $\KLukProdProb$ with the exception that $\Imc_\Mfrak(\Prob(p),w)=v(p,w)$. 
%
\begin{restatable}{proposition}{probandfuzzy}\label{prop:probandfuzzy}
For finite $\Gamma \cup \{ \chi \} \subseteq \Lmodal$, $\Gamma \models \chi$ iff $\Gamma \models_\Prob \chi$.  
\end{restatable}
\begin{proof}
If $\Gamma \not\models_\Prob \chi$, then take any countermodel $\Mfrak$ and construct a $\KLukProd$-model $\Mfrak'$ by $v(p, w) := \Imc_\Mfrak (\Prob(p), w)$. It is clear that $\Imc_\Mfrak (\phi, w) = \Imc_{\Mfrak'}(\phi, w)$ for all $\phi \in \Lmodal$ and so $\Gamma \not\models \chi$. Conversely, if $\Gamma \not\models \chi$, then take any countermodel $\Mfrak$ and define an $SI$-model $\Mfrak'$ where $X=2^{\Prop[\Gamma\cup\{\chi\}]}$ 
and where $\mu_w$ is determined by $\mu_w (\{ p \})=v(p,w)$. Moreover, let $E(p) = \{ p \}$ if $p$ occurs in $\Gamma \cup \{ \chi \}$ and let $E(p) = \varnothing$ otherwise. It is clear that $\Imc_\Mfrak (\phi, w) = \Imc_{\Mfrak'}(\phi, w)$ for all $\phi \in \Lmodal$ such that $\phi$ is a~subformula of a~formula occurring in $\Gamma \cup \{ \chi \}$. Hence, by the obvious extension of Theorem~\ref{theorem:semanticsequivalence} to local consequence, $\Gamma \not\models_\Prob \chi$.
\end{proof}

\section{Expressivity}\label{sec:express}

Note that the truth degree of an implication $\phi\rightarrow\psi$ is inversely proportional to the truncated difference between the truth degrees of $\phi$ and $\psi$.\footnote{In particular, $\Imc_\Mfrak(\phi \to \psi, w) = 1 - \max \{ \Imc_\Mfrak (\phi, w) - \Imc_\Mfrak (\psi, w), 0\}$.} Thus, $\phi \rightarrow \psi$ can be read as ``The truth degree of $\phi$ is \emph{not much higher} than the truth degree of $\psi$''. If the truth degree of $\phi$ is less or equal to the truth degree of $\psi$, 
then the truth degree of $\phi \rightarrow \psi$ is~$1$. If, on the other hand, the truth degree of $\phi$ is higher than the truth degree of $\psi$, 
then the truth degree of $\phi \rightarrow \psi$ is $1$ minus the difference. Similarly, the truth degree of $\phi\leftrightarrow\psi$ 
is inversely proportional to the \emph{absolute difference} between the values of $\phi$ and $\psi$. Consequently, we can read $\phi \tot \psi$ as ``The truth degree of $\phi$ does not \emph{differ much} from the truth degree of $\psi$''.

Recall that, given a~set $\Pmc$ of probability measures on a~Boolean algebra $X$, the \emph{lower} and \emph{upper probability} of $x \in X$, respectively, are defined as $\Pmc_*(x) = \inf \{ \mu(x) \mid \mu \in \Pmc \}$ and $\Pmc^*(x) = \sup \{ \mu(x) \mid \mu \in \Pmc \}$ 
(see~\cite{HalpernPucella2002} and~\cite[Ch.~2.3]{Halpern2017}). Lower and upper probabilities are useful for expressing qualitative uncertainty. For example, the interval $[\Pmc_*(x), \Pmc^*(x)]$ can be seen as putting lower and upper bounds on an agent's ignorance concerning the probability of $x$. A crucial feature of the many-valued modal framework is that lower and upper probabilities are directly expressed by formulas of the language. Given a~model $\Mfrak$ and a~world $w$, the modal formula $\Box_a \Prob (\alpha)$ can be seen as a probabilistic term expressing the lower probability of $\alpha$ with respect to $\Pmc = \{ \mu_{w'} \mid w' \in R_a(w) \}$ since \[ \Imc_\Mfrak (\Box_a \Prob (\alpha), w) = \inf \{ \Imc_\Mfrak (\Prob (\alpha), w') \mid w' \in R_a(w) \} =  \inf \{ \mu_{w'}(E(\alpha)) \mid w' \in R_a(w) \} = \Pmc_*(E(\alpha)) \, .\] 
Similarly, $\Diamond_a \Prob (\alpha)$ expresses the upper probability of $\alpha$. It follows that, on the epistemic interpretation of $a \in \Amsf$, the formula $\Box_a \Prob (\alpha) \tot \Diamond_a \Prob (\alpha)$ expresses the amount of qualitative certainty of agent $a$ concerning $\alpha$. In particular, the truth degree of $\Box_a \Prob (\alpha) \tot \Diamond_a \Prob (\alpha)$ is inversely proportional to the size of the interval $[\Pmc_*(E(\alpha)), \Pmc^*(E(\alpha))]$. Let us write $\mathsf{Cert}_a(\alpha)$ instead of $\Box_a \Prob (\alpha) \tot \Diamond_a \Prob (\alpha)$. We can use {\L}ukasiewicz implication to express comparisons between the amounts of qualitative uncertainty of two agents. For instance, $\mathsf{Cert}_a(\alpha) \to \mathsf{Cert}_b(\alpha)$ says that $a$ is not much more certain about $\alpha$ than $b$.\footnote{A similar interpretation applies, of course, if we read $a \in \Amsf$ as actions. Then $\mathsf{Cert}_a(\alpha)$ expresses the extent to which action $a$ ``fixes'' the probability of $\alpha$.}

\begin{remark}
Halpern and Pucella \cite{HalpernPucella2002} introduce a propositional probabilistic logic suitable for reasoning about upper probabilities. A Halpern-Pucella model is, using our notation, a tuple $\langle W, S, \Pmc, E \rangle$ where $W$ is a set, $S$ is a Boolean subalgebra of $2^W$ and $E$ is a function from a set of propositional variables to $S$~--~similarly as in our models~--~and $\Pmc$ is a set of probability measures on $S$. The language Halpern and Pucella use is a variant of the non-modal linear inequality language of \cite{FaginHalpernMegiddo1990} with a primitive upper probability operator. Halpern and Pucella provide a sound and complete axiomatisation of their logic and show that the corresponding satisfiability problem is $\np$-complete. Marchioni \cite{Marchioni2008} introduces a propositional many-valued logic for reasoning about upper probabilities in semantic structures derived from Halpern-Pucella models.\footnote{Instead of a set of probability measures $\Pmc$, Marchioni uses an upper probability measure $\pi : S \to [0,1]$, that is, a function that satisfies certain conditions that allow it to represent an upper probability given by a set of probability measures. The reader is referred to \cite{HalpernPucella2002,Marchioni2008} for details.} The logic is based on Rational {\L}ukasiewicz logic \cite{Gerla2001}  and it does not contain modal operators. Marchioni provides a sound and (finitely) strongly complete axiomatisation of his logic, and he shows that the corresponding satisfiability problem is $\np$-complete.

Our framework can be seen as a generalisation of both \cite{HalpernPucella2002} and \cite{Marchioni2008}. Halpern-Pucella models are a~$\mathsf{KD45}$-like version of our $\mathsf{K}$-like probabilistic models. Instead of having a fixed set of ``globally'' accessible states giving rise to $\Pmc$, we introduce an accessibility relation which expresses the idea that each state has its own set of ``locally'' accessible states. This difference is well reflected in the comparison of the computational complexity: as shown below, deduction problems in our framework are $\pspace$-complete as opposed to $\conp$-complete in the case of \cite{HalpernPucella2002} and \cite{Marchioni2008}. Another difference is in the language we are using. Instead of a primitive upper/lower probability operator, we use a combination $\Diamond \Prob / \Box \Prob$ of a modal operator and the probability operator $\Prob$.\footnote{In addition, our syntax and semantics are both multi-modal, allowing for modelling multiple agents or actions.} The same syntactic approach to representing lower probabilities is taken by Corsi et al.~\cite{CorsiFlaminioGodoHosni2023}. Their semantics is a special case of our probabilistic models with only one equivalence accessibility relation (mono-modal $\mathsf{S5}$-style probabilistic models). On the other hand, Corsi et al.~\cite{CorsiFlaminioGodoHosni2023} work with a richer language that allows for the inclusion of $\Box$ within the scope of $\Prob$. They demonstrate that their modal probabilistic framework can formalise reasoning about probability, belief functions, and lower probabilities within a unified framework. The main results of \cite{CorsiFlaminioGodoHosni2023} are related to complete axiomatisations, and complexity is not studied.
\end{remark}

We can also use {\L}ukasiewicz implication to express how actions modify the probability of certain events. For example, $\Prob (\alpha) \to \Box_a \Prob (\alpha)$ says that action $a$ \emph{cannot decrease} the probability of $\alpha$ \emph{much}~--- the truth degree in $w$ is $1$ if $\mu_w(E(\alpha)) \leq \inf \{ \mu_{w'}(E(\alpha)) \mid w' \in R_a(w)\}$, that is, the probability of $\alpha$ cannot decrease at all, and it is $1$ minus the difference otherwise. Similarly, $\Diamond_a \Prob (\alpha) \to \Prob (\alpha)$ says that $a$ cannot \emph{increase} the probability of $\alpha$ much. We can write $\mathsf{NoDec}_a(\alpha)$ instead of $\Prob (\alpha) \to \Box_a \Prob (\alpha)$ and, as before, use formulas of the form $\mathsf{NoDec}_a(\alpha) \to \mathsf{NoDec}_b(\alpha)$ to express comparisons between actions $a$ and $b$. Similarly for $\Diamond_a \Prob (\alpha) \to \Prob (\alpha)$ written as $\mathsf{NoInc}_a(\alpha)$. We reiterate that these comparisons are not precise statements (true or false) but rather imprecise~--~that is, true to a~degree. 

It is well known that in the presence of the constant $\onehalf$ and the ``arithmetic'' operators $\bullet$ and $\rightarrow_\Pi$, one can define constants $\qmbf$ for all rational numbers $q \in [0,1] \cap \Qmbb$ in {\L}ukasiewicz logic. In combination with {\L}ukasiewicz $\to$ and $\tot$, this allows us to express comparisons with rational thresholds. For example, the formula $\mathbf{0.7} \to \Diamond_a \Prob (\alpha)$ read in terms of actions says that the supremum (``best case'') of the possible probabilities of $\alpha$ resulting from action $a$ is not much lower than $0.7$. The definable operator $\triangle$ turns imprecise statements into precise ones. For instance, $\triangle (\mathbf{0.7} \to \Prob (\alpha))$ has truth degree $1$ if probability of $\alpha$ is at least $0.7$ and has truth degree $0$ otherwise. This means that we may define a~restricted version of the ``Markovian'' modalities of \cite{Aumann1999a,HeifetzMongin2001,KozenEtAl2013} etc.: $L_\qmbf(\alpha) \coloneqq \triangle (\qmbf \to \Prob (\alpha))$ is true (has truth degree $1$) in states where the probability of $\alpha$ is at least $q \in [0,1] \cap \Qmbb$ and false (has truth degree $0$) otherwise.\footnote{This is only a~restricted version since $L_\qmbf$ cannot be nested in our language.} In addition, if $a$ is interpreted epistemically, then $\Box_a L_\qmbf(\alpha)$ expresses a~crisp statement about epistemic attitudes towards probabilities~--- namely, ``Agent $a$ believes that the probability of $\alpha$ is at least $q$''.

We observe briefly that just as in the propositional case, constants do not change the expressivity of the language. The following statement generalises~\cite[Lemmas~3.3.11 and~3.3.13]{Hajek1998}.
\begin{restatable}{proposition}{rationalconstantselimination}\label{prop:rationalconstantselimination}
Let $\Gamma\subseteq\LProb$ be finite and~$n$ be the least common denominator of the constants in~$\Gamma$. Then there is $\Gamma'$ without constants s.t.\ $\Gamma$ is satisfiable iff $\Gamma'$ is and $\Gamma'$ is of polynomial size w.r.t.\ $\max(\lmc[\Gamma],n)$.
\end{restatable}
\begin{proof}
As $\Gamma$ is finite, $n$~exists. We can express $\sfrac{1}{n}$ as follows. For fresh $\{q_1,\ldots,q_n,q\}\subseteq\Prop$, we set
\begin{align*}
\mathbf{\sfrac{1}{n}}(q)&\coloneqq(\underbrace{\Prob(q)\oplus\ldots\oplus\Prob(q)}_{n-1\text{ times}})\leftrightarrow\neg\Prob(q)&\mathbf{\sfrac{m}{n}}(q_m)&\coloneqq \Prob(q_m)\leftrightarrow(\underbrace{\Prob(q)\oplus\ldots\oplus\Prob(q)}_{m\text{ times}})
\end{align*}
It is clear that $\Imc(\mathbf{\sfrac{1}{n}}(q),w)=1$ iff $\Imc_\Mfrak(\Prob(q),w)=\sfrac{1}{n}$ and $\Imc_\Mfrak(\mathbf{\sfrac{m}{n}}(q_m),w)=1$ iff $\Imc(\Prob(q_m),w)=\sfrac{m}{n}$ (and $v(q)=\sfrac{1}{n}$). Now, given $\phi\in\LProb$, we define $\phi^\flat$ as follows:
\begin{align*}
\Prob(\alpha)^\flat&=\Prob(\alpha)&\mathbf{\sfrac{m}{n}}^\flat&=\Prob(q_m)&
(\heartsuit\phi)^\flat&=\heartsuit\phi^\flat&(\phi\circ\chi)^\flat&=\phi^\flat\circ\chi^\flat\tag{$\heartsuit\in\{\neg,\triangle,\Box\}$, $\circ\in\{\rightarrow,\rightarrow_\Pi,\bullet\}$}
\end{align*}
Now let $\dmc[\Gamma]=\max\{\dmc(\phi)\mid\phi\in\Gamma\}$ and define
\begin{align*}
\Gamma^\flat&=\{\phi^\flat\mid\phi\in\Gamma\}\cup\{\Box^i\mathbf{\sfrac{1}{n}}(q)\mid i\in\{0,\ldots,\dmc[\Gamma]\}\}\cup\bigcup\limits^{\dmc[\Gamma]}_{i=0}\{\Box^i\mathbf{\sfrac{m}{n}}(q_m)\mid\mathbf{\sfrac{m}{n}}\text{ occurs in~}\Gamma\}
\end{align*}
It is clear that $\Gamma^\flat$ is only polynomially larger than $\Gamma$ when $\lmc(\Gamma)\geq n$.

Now let $\Mfrak$ be a~$\KLukProdProb$-model and $w\in\Mfrak$ s.t.\ $\Imc_\Mfrak(\phi,w)=1$ for each $\phi\in\Gamma$. By Proposition~\ref{prop:simplification}, we may w.l.o.g.\ assume that $\Mfrak=\langle W,R,2^{2^{\Prop[\Gamma]}},\mu,E^\Cmsf\rangle$ is an SI-model over $2^{2^{\Prop[\Gamma]}}$. We define an SI-model $\Mfrak^\flat\!=\!\langle W,R,2^{2^{\Prop[\Gamma^\flat]}},\mu^\flat,{E^\flat}^\Cmsf\rangle$ s.t.\ $\mu^\flat_w({E^\flat}^\Cmsf(\alpha))\!=\!\mu_w(E^\Cmsf(\alpha))$ for each $\alpha$ s.t.\ $\Prop(\alpha)\!\cap\!\{q_1,\ldots,q_n,q\}\!=\!\varnothing$, $\mu^\flat_w(E(q_m))=\frac{m}{n}$ for every $q_m$, and $\mu^\flat_w(E(q))=\frac{1}{n}$. It is now clear that $\Imc_{\Mfrak^\flat}(\Prop(q_m),w)=\frac{m}{n}$ in every $w\in\Mfrak$. Thus, as values of formulas depend only on states not further than their modal depth we have $\Imc_\Mfrak(\phi,w)=\Imc_{\Mfrak^\flat}(\phi^\flat,w)$ for every $\phi\in\Gamma$. The result follows.
\end{proof}

In \cite{HajekEtAl2000}, $\to_\Pi$ is used to express conditional probabilities $\Prob (\alpha \mid \beta) \coloneqq \Prob (\beta) \to_\Pi \Prob (\alpha \land \beta)$. It follows from this definition that $\Prob (\alpha \mid \beta)$ has truth degree $1$ if the probability of $\beta$ is $0$. 
A many-valued probabilistic logic with primitive formulas expressing conditional probabilities was proposed in \cite{GodoMarchioni2006}. Flaminio~\cite{Flaminio2007} shows that reasoning about conditional probabilities can be carried out in a~probabilistic logic based on an extension of Rational {\L}ukasiewicz logic (R\L{}) with the $\triangle$ operator, where the satisfiability problem is $\np$-complete. Our reason to base our logic on $\Luk\bm{\Pi}\onehalf$ instead of R\L{} with $\triangle$ is that the former can express additional statements about probability, such as independence statements. 

We now turn to linking the expressive resources of our language to the examples discussed in Section~\ref{sec:frames}. Example~\ref{exam:robot1} can be formalised using $L, I, S \in \Prop$, expressing ``light on'', ``item retrieved'' and ``start state'', respectively. It is clear how the event function $E$ should be defined on the frame depicted in Figure~\ref{fig:robot-example}~--- for instance, $E(\neg S \land I) = \{ e_{L \land I}, e_{\neg L \land I} \}$. The formula $\Prob (\neg S \land I)$ expresses the probability of the event $E(\neg S \land I)$, ``the robot will retrieve the correct item'', and the formula $\Box_a \Prob (\neg S \land I)$ expresses the probability of this event after the execution of action $a$, that is, after switching the light switch. In each start state, the action $a$ has exactly one output state --- $s_L$ is switched to $s_{\neg L}$ and vice versa. The formula $\Prob (\neg S \land I) \to \Box_a \Prob (\neg S \land I)$ says that switching the light switch cannot decrease the probability of the robot retrieving the correct item \emph{much}. In state $s_L$, the truth degree of this formula is $0.4$ since the truth degree of $\Prob (\neg S \land I)$ in $s_L$ is $0.8$ and in $s_{\neg L}$ it is $0.2$. On the other hand, the truth degree of $\Prob (\neg S \land I) \to \Box_a \Prob (\neg S \land I)$ in $s_{\neg L}$ is $1$. The formula $\Prob (\neg S \land I) \tot \Box_a \Prob (\neg S \land I)$, saying that the probability of the robot retrieving the correct item is not \emph{changed much} by $a$, has truth degree $0.4$ in both start states $s_L$ and $s_{\neg L}$.

Example~\ref{exam:MarkovChain} shows that finite discrete-time Markov chains can be seen as probabilistic frames. Take a~Markov chain $\Mmbf=\langle N,\pfrak,\qfrak)$ where $N=\{ 1,\ldots, n \}$, and consider an extension of $\Ffrak_\Mmbf$ with an event function $E$ such that $E(\alpha_i) = \{ i \}$ for $i \in N_0$ and a~tuple of Boolean formulas $\alpha_i$. The probability of a~path $P(i_0, \ldots, i_m) = \qfrak(i_0) \cdot \prod_{j = 1}^m\pfrak(i_{j-1}, i_j)$ corresponds to the truth degree of the formula
\[ \Prob(i_0, \ldots, i_m) \coloneqq \Prob (\alpha_{i_0}) \bullet \Diamond_{i_1} (\Prob (\alpha_{i_1}) \bullet \Diamond_{i_2} ( \ldots (\Diamond_{i_{m-1}} \Prob (\alpha_{i_m}) ) \ldots ))\] 
in the state $0$.\footnote{Note that, for all $m>0$, the power set of the set of all paths of length $m$ forms a finite Boolean algebra. This justifies the notation $P(i_0, \ldots, i_m)$.} For instance, the probability of $P(1, 2, 3)$ is the truth degree of $\Prob(\alpha_1) \bullet \Diamond_1 (\Prob(\alpha_2) \bullet \Diamond_2 \Prob(\alpha_3))$. As before, we can use the resources of our language to express comparisons between probabilities of various paths, e.g.,~$\Prob(1, 2, 3) \to \Prob(1, 3, 2)$ or $\triangle(\Prob(1, 2, 3) \to \Prob(1, 3, 2))$. 

In practical applications, one very often encounters \emph{finitely branching} frames, that is, probabilistic frames where $R_a(w)$ is finite for all $a \in \Amsf$ and all $w \in W$. For example, take frames where $R_a$ is a partial function for all $a \in \Amsf$, representing deterministic actions.\footnote{Finitely branching Kripke frames are used in the semantics of Strict Deterministic Propositional Dynamic Logic \cite{HalpernReif1983}, for example. Finitely branching probabilistic frames provide semantics for a fragment of a probabilistic version of this logic.} Another obvious example are finite frames where worlds are individuated by values of a finite set of discrete random variables. We write $\models^\fb_\Prob$ and $\models^\fb$ to denote the entailment (satisfaction) relations over the class of finitely-branching frames in $\KLukProdProb$ and $\KLukProd$, respectively. 

\section{Decidability and complexity}\label{sec:results}
Let us now tackle the complexity of reasoning in $\KLukProdProb$. We will consider two fundamental tasks. First, \emph{model-checking} --- given a~$\KLukProdProb$-model $\Mfrak$, a~state $w\in\Mfrak$, and a~formula $\phi$, determine the truth degree of $\Imc_\Mfrak(\phi,w)$. The second task is \emph{finitary entailment} over finitely branching models --- given a~finite $\Gamma\cup\{\chi\}$, decide whether $\Gamma\models^\fb_\Prob\chi$.

As expected, model-checking can be done in polynomial time.
\begin{restatable}{theorem}{modelchecking}\label{theorem:modelchecking}
Given a~finite model $\Mfrak=\langle W,R,S,\mu,E\rangle$, $w\in W$, and $\phi\in\LProb$, it takes polynomial time w.r.t.\ $|W|+|S|+\lmc(\phi)$ to determine the truth degree of~$\Imc_\Mfrak(\phi,w)$.
\end{restatable}
\begin{proof}
We proceed by induction on $\phi$. For a~basis case of $\phi=\Prob(\alpha)$, we first determine $\|\alpha\|$ which takes polynomial time w.r.t.\ $\lmc(\phi)+|S|$. Then, using the definition of $\mu_w$, we calculate $\mu_w(\|\alpha\|)$ which also takes polynomial time. The cases of propositional connectives can be tackled by straightforward application of the induction hypothesis. Finally, to determine $v(\Box\chi,w)$, we calculate $v(\chi,w')$ in every $w'\in R(w)$, which takes polynomial time by the induction hypothesis. Then, we take the minimal among those, which gives us the result.
\end{proof}




Let us now tackle the complexity of $\KLukProdProb$-entailment over finitely-branching frames. We begin with the presentation of tableaux for $\KLukProd_\fb$ --- the logic $\KLukProd$ over finitely-branching frames.
\begin{definition}[Constraints]\label{def:tableauxlanguage}
We fix a~countable set $\Wmsf=\{w,w',w_0,w_1,\ldots\}$ of \emph{state labels} and define:
\begin{itemize}[noitemsep,topsep=2pt]
\item a~\emph{labelled formula} to be $w:\phi$ with $w\in\Wmsf$ and $\phi\in\LProb$;
\item a~\emph{formulaic constraint} as $w:\phi\triangledown i$ s.t.\ $w:\phi$ is a~labelled formula, $\triangledown\in\{\leq,\geq\}$ and $i$~a~polynomial;
\item a~\emph{numerical constraint} as $w:i\triangleleft j$ with $w\in\Wmsf$, $\triangleleft\in\{\leq,<\}$, and~$i$ and~$j$ polynomials.
\end{itemize}
\end{definition}
\begin{definition}[$\TKLukProd$ --- constraint tableaux for $\KLukProd_\fb$]\label{def:KLuktableaux}
A~\emph{constraint tableau} is a~downward branching tree of \emph{formulaic constraints}, \emph{numerical constraints}, and \emph{relational terms} of the form $wR_aw'$ ($w,w'\in\Wmsf$). A~branch can be extended by an application of a~rule from Fig.~\ref{fig:tableauxrules}. Given a~tableau branch $\Bmc=\Fmc\cup\Nmc$, we let $\Wmsf_\Bmc=\{w\mid w\in\Wmsf,\text{ and }w\text{ occurs in }\Bmc\}$ and set
\begin{align*}
\Fmc&=\{w_1:\phi_1\triangledown i_1,\ldots,w_m:\phi_m\triangledown i_m\}&\Nmc&=\{w'_1:i'_1\triangledown j'_1,\ldots,w'_n:i'_n\triangledown j'_n,w''_1:\onehalf\triangledown j_1,\ldots,w''_r:\onehalf\triangledown j_r\}
\end{align*}
We define $\Fmc^\Imbb_w=\{x_{w:\phi}\triangledown i\mid w:\phi\triangledown i\in\Fmc\}$ and $\Nmc^\Imbb_w=\{x\triangledown y\mid w:(x\triangledown y)\in\Nmc\}$ for every $w\in\Wmsf_\Bmc$ and set $\Bmc^\Imbb_w=\Fmc^\Imbb_w\cup\Nmc^\Imbb_w$. Furthermore, we set $\Fmc^\Imbb=\bigcup_{w\in\Bmc_\Wmsf}\Fmc^\Imbb_w$, $\Nmc^\Imbb=\bigcup_{w\in\Bmc_\Wmsf}\Nmc^\Imbb_w$, and $\Bmc^\Imbb=\bigcup_{w\in\Bmc_\Wmsf}\Bmc^\Imbb_w$.

A~tableau branch $\Bmc$ is called \emph{open} if its corresponding system of inequalities $\Bmc^\Imbb$ has a~solution over $[0,1]$ and is \emph{closed} otherwise. An open branch is \emph{complete} if its conclusion is also present on the branch for every premise of a~rule on the branch. A~tableau is closed when all its branches are closed.
\end{definition}
\begin{figure}
\centering
\begin{align*}
\neg_\leq:\dfrac{w:\neg\phi\leq i}{w:\phi\geq1-i}&&
\rightarrow_\leq:\dfrac{w:\phi\rightarrow\chi\leq i}{i\geq1\left|\begin{matrix}w:\phi\geq1-i+j\\w:\chi\leq j\\w:j\leq i\end{matrix}\right.}&&
\Box_\leq:\dfrac{w:\Box_a\phi\leq i}{\begin{matrix}wR_aw'\\w':\phi\leq i\end{matrix}}\\[.25em]
\neg_\geq:\dfrac{w:\neg\phi\geq i}{w:\phi\leq1-i}&&
\rightarrow_\geq:\dfrac{w:\phi\rightarrow\chi\geq i}{\begin{matrix}w:\phi\leq1-i+j\\w:\chi\geq j\end{matrix}}&&
\Box_\geq:\dfrac{\begin{matrix}w:\Box_a\phi\geq i\\wR_au\end{matrix}}{u:\phi\geq i}\\[.25em]
\bullet:\dfrac{w:\phi\bullet\chi\triangledown i}{\begin{matrix}w:\phi=j_1\\w:\chi=j_2\\w:j_1\cdot j_2\triangledown i\end{matrix}}
&&
{\rightarrow_\Pi}_\leq:\dfrac{w:\phi\rightarrow_\Pi\chi\leq i}{i\geq1\left|\begin{matrix}w:\phi=j_1\\w:\chi=j_2\\w:j_2/j_1\leq i\end{matrix}\right.}
&&
{\rightarrow_\Pi}_\geq:\dfrac{w:\phi\rightarrow_\Pi\chi\geq i}{\begin{matrix}w:\phi\geq j\\w:\chi\leq j\end{matrix}\left|\begin{matrix}w:\phi=j_1\\w:\chi=j_2\\w:j_2/j_1\geq i\end{matrix}\right.}
\end{align*}
\caption{Tableaux rules: $j$ is a~new variable; $w'$ is fresh on the branch and $u$ is present on the branch; $w:\psi=j$ is a~shorthand for $\{w:\psi\geq j,w:\psi\leq j\}$.}
\label{fig:tableauxrules}
\end{figure}
\begin{definition}[Realising model]\label{def:realisingmodel}
Let $\Bmc$ be a~complete open branch, $\Bmc^\Imbb$ its corresponding system of inequalities, $\Smc_{\Bmc^\Imbb}$ its solution, and given a~variable $x_{w:\phi}$, let $\Smc_{\Bmc^\Imbb}(x_{w:\phi})$ denote its value in $\Smc_{\Bmc^\Imbb}$. We say that a~\emph{$\KLukProd$-model} $\Mfrak=\langle W,R,v\rangle$ \emph{realises} $\Smc_{\Bmc^\Imbb}$ if $W=\{w\mid w\text{ occurs on }\Bmc\}$, $R_a=\{\langle w,w'\rangle\mid wR_aw'\in\Bmc\}$, and $\Imc_\Mfrak(\phi,w)=\Smc_{\Bmc^\Imbb}(x_{w:\phi})$ for each $\phi$. In particular, if $w:\phi\triangledown i\in\Bmc$ and $\Imc_\Mfrak(\phi,w)=\Smc_{\Bmc^\Imbb}(x_{w:\phi})$, we say that $\Mfrak$ realises $w:\phi\triangledown i$. We say that $\Mfrak$ \emph{realises $\Bmc$} if there is some solution of $\Bmc^\Imbb$ that is realised by~$\Mfrak$.
\end{definition}

Observe from the definition above that to construct a~full countermodel, it suffices to have \emph{one} complete open branch. Let us show the completeness of~$\TKLukProd$. The next theorem can be proved in a~standard manner.
\begin{restatable}{theorem}{TKLukcompleteness}\label{theorem:TKLukcompleteness}
Let $\Gamma\cup\{\chi\}\subseteq\LProb$ be finite. Then $\Gamma\models^\fb\chi$ iff the tableau beginning with $\{w:\phi\geq1\mid\phi\in\Gamma\}\cup\{w:\chi\leq d,d<1\}$ is closed.
\end{restatable}
\begin{proof}
For the sake of simplicity, we assume that our language is mono-modal. The proof can be carried out in the same manner when $|\Amsf|\geq2$.

For the soundness part, we observe first that closed branches do not have realising models. Now, it suffices to show that if $\Mfrak$ realises the premises of a~rule, it also realises the conclusion. For propositional rules, this follows straightforwardly from Definition~\ref{def:KLukProbmodels}. Let us now consider modal rules. Assume that $\Mfrak$ realises $w:\Box\phi\leq i$. Then $\Imc(\Box\phi,w)\leq i$ and there must be some $w'\in R(w)$ s.t.\ $\Imc(\phi,w')\leq i$. Hence, the conclusion is also realised. The rule $\Box_\geq$ can be dealt with similarly.

For the completeness part, we show that every complete open branch $\Bmc$ has a~realising model. Let $\Smc_{\Bmc^\Imbb}$ be a~solution of the system of inequalities $\Bmc^\Imbb$ that corresponds to~$\Bmc$. We construct the model $\Mfrak=\langle W,R,v\rangle$ as follows: $W=\{w\mid w\text{ occurs on }\Bmc\}$, $R=\{\langle w,w'\rangle\mid wRw'\in\Bmc\}$, $v(p,w)=\Smc_{\Bmc^\Imbb}(x_{w:p})$. It remains to show that $\Imc_\Mfrak(\chi,w)=\Smc_{\Bmc^\Imbb}(x_{w:\chi})$ for every $\chi$ occurring on~$\Bmc$.

We proceed by induction on~$\chi$. The base case of $\chi\in\Prop$ holds by construction of~$\Mfrak$. The propositional cases can be verified by a~straightforward application of the induction hypothesis. Let us now consider the case of $\chi=\Box\psi$. Assume that $w:\Box\psi\leq i\in\Bmc$. As $\Bmc$ is complete, $\{wRw',w':\psi\leq i\}\subseteq\Bmc$. By the induction hypothesis, $\Mfrak$~realises $w':\psi\leq i$. Thus, there is $w'\in R(w)$ s.t.\ $\Imc_\Mfrak(\psi,w')\leq i$. It follows that $\Imc_\Mfrak(\Box\psi,w)\leq i$, as required. Now assume that $w:\Box\psi\geq i\in\Bmc$. If there is no $u$ s.t.\ $wRu\in\Bmc$, it means that $R(w)=\varnothing$, whence $\Imc_\Mfrak(\Box\psi,w)=1$. Thus, $w:\Box\psi\geq i$ is realised. If there is some $u$ s.t.\ $wRu\in\Bmc$, then for each such $u$, $u:\psi\geq i\in\Bmc$. Again, by the induction hypothesis, all these constraints are realised, whence $\Imc_\Mfrak(\Box\psi,w)\leq i$, as required.
\end{proof}

We are finally ready to establish the computational complexity of $\models^\fb_\Prob$. We adapt the technique from~\cite{HajekTulipani2001} and combine it with a~standard decision algorithm for $\Kmbf$ from~\cite{BlackburndeRijkeVenema2010}.
\begin{restatable}{theorem}{KLulmultProbpspace}\label{theorem:KLulmultProbpspace}
Let $\Gamma\cup\{\chi\}\subseteq\LProb$ be finite. Then it is $\pspace$-complete to decide whether $\Gamma\models^\fb_\Prob\chi$.
\end{restatable}
\begin{proof}
Let us begin with the $\pspace$-hardness. We provide a~reduction from the validity in $\Kmbf$, which is known to be $\pspace$-complete. Namely, let $\phi$ be a~formula over $\{\neg,\rightarrow,\Box\}$ and let further, $\phi^\triangle$~be the result of replacing each variable $p$ occurring in~$\phi$ with $\triangle\Prob(p)$. Now observe that $\triangle$-formulas have values in $\{0,1\}$ and $\KLukProdProb$-connectives behave classically on $\{0,1\}$. As $\Kmbf$ is complete over finitely branching frames, $\phi$ is $\Kmbf$-valid iff $\phi^\triangle$ is $\KLukProdProb$-valid on finitely branching frames.

Let us now provide a~$\pspace$ decision procedure. We assume that $\Xi=\Gamma\cup\{\chi\}$, that there are $r$~different probabilistic atoms in~$\Xi$, and $\Prop[\Xi]=\{p_1,\ldots,p_m\}$. Observe that tableaux for finite sets of formulas terminate since rules have the branching factor of at most~$2$ and decompose formulas. Consider the following set: $\Tmc(\Xi)=\{w^0_1\!:\!\chi\!\leq\!c,c\!<\!1\}\cup\{w^0_1:\phi\geq1\mid\phi\in\Gamma\}$. By~\cite[Theorem~3.3]{Canny1988}, it is in $\pspace$ to determine whether the given branch is closed as $\Bmc^\Imbb_w$'s are systems of \emph{polynomial} inequalities over~$\Rmbb$.

The algorithm works in two stages. First, we build a~$\KLukProd$-tableau in a~depth-first manner. Second, we recall that by Proposition~\ref{prop:simplification}, we can assume that all $\alpha$'s are evaluated on $2^{2^{\Prop(\phi)}}$. At this stage, we check that the values of modal atoms are coherent as a~measure on $S=2^{2^{\Prop[\Xi]}}$, i.e., that there is a~probability measure $\mu$ s.t.\ $\mu_w(E(\alpha))=x_{w:\Prob(\alpha)}$ for every modal atom $\Prob(\alpha)$ and state label~$w$ in~$\Xi$.

We begin by applying propositional rules to $\Tmc(\Xi)$, guessing the branch if needed and checking whether the resulting system of inequalities has solutions. Once there are no formulas in $w^0_1$ with a~propositional principal connective, we pick a~constraint $w^0_1:\Box\psi\leq i$ and apply the $\Box_\leq$ rule to it. This produces a~new state $w^1_{1,1}$\footnote{Here, the upper index denotes the ‘depth’ of the state on the branch; the lower left index stands for the position of the state on the level; the lower right index identifies the position of the predecessor on its level.} and a~relational term $w^0_1Rw^1_{1,1}$. Now we can apply rules $\Box_\geq$ using this new state. In this new state, we apply propositional rules until no formulas with principal propositional connectives remain. We proceed this way until we reach a~state, say $w^n_{1,1}$ generated from $w^{n-1}_{1,1}$, which does not contain modal formulas. Observe that $n\in\Omc(\lmc[\Xi])$ because the length of the branch in the model is bounded by the maximal number of nested modalities in a~formula of $\Xi$.

Once formulas in $w^n_{1,1}$ are decomposed into modal atoms, we have the following system of polynomial inequalities that corresponds to the current branch~$\Bmc$ and state~$w^n_{1,1}$ ($\psi$'s are formulas in~$w^n_{1,1}$):
\begin{align}\label{equ:initialsystemmult}
x_{w^n_{1,1}:\psi_1}\triangledown l'_1&&\ldots&&x_{w^n_{1,1}:\psi_k}\triangledown l'_k&&x_{w^n_{1,1}:\Prob(\alpha_1)}\triangledown l_1&&\ldots&&x_{w^n_{1,1}:\Prob(\alpha_r)}\triangledown l_n&&j_1\triangleleft j'_1&&\ldots&&j_{r'}\triangleleft j'_{r'}\tag{$\Bmc^\Imbb_{w^n_{1,1}}$}
\end{align}

We now need to check that the values of $x_{w^n_{1,1}:\Prob(\alpha)}$'s are coherent as a~probability measure. To do this, we consider the following system of \emph{linear} inequalities for every $i\in\{1,\ldots,r\}$:
\begin{align}\label{equ:bigsystemmult}
\sum\limits_{e\in\{0,1\}^m}u_e&=1&\sum\limits_{e\in\{0,1\}^m}a_{i,e}\cdot u_e&=x_{w^n_{1,1}:\Prob(\alpha_i)}\tag{$\exp\Bmc^\Imbb_{w^n_{1,1}}$}
\end{align}
Here, the $e$ indices represent the subsets of $\Prop[\Xi]$, i.e., the atoms of $2^{2^{\Prop[\Xi]}}$: $a_{i,e}=1$ when $e\models\alpha_i$ (this takes polynomial time to establish), and the value of $u_e$ is the measure assignment to~$e$. For example, if there are only two variables: $p$ and $q$, then $u_{01}$ corresponds to the measure of $\{q\}$ and $u_{11}$ to the measure of $\{p,q\}$. Thus, \eqref{equ:bigsystemmult} has solutions iff the values of $x_{w^n_{1,1}:\Prob(\alpha_i)}$'s are coherent as a~measure.

Now observe that $u_e$'s do not occur in~\eqref{equ:initialsystemmult} and that~\eqref{equ:bigsystemmult} has only $r+1$ lines. Thus, even though there are $2^{|\Prop[\Xi]|}$ variables $u_e$, we can guess a~list $L$ of $r+1$ words $e$~over $\{0,1\}$ using~\cite[Lemma~2.5]{FaginHalpernMegiddo1990} and let $u_e=0$ for each $e\notin L$. This produces the following system of inequalities.
\begin{align}\label{equ:smallsystemmult}
\sum\limits_{e\in L}u_e&=1&\sum\limits_{e\in L}a_{i,e}\cdot u_e&=x_{w^n_{1,1}:\Prob(\alpha_i)}\tag{$\mathrm{poly}\Bmc^\Imbb_{w^n_{1,1}}$}
\end{align}
Clearly, \eqref{equ:smallsystemmult} is of polynomial size w.r.t.\ $\lmc[\Xi]$ and has solutions iff~\eqref{equ:bigsystemmult} does. We can now treat \eqref{equ:smallsystemmult} and~\eqref{equ:initialsystemmult} as quantifier-free formulas of the real field theory and consider their conjunction. By~\cite[Theorem~3.3]{Canny1988}, it is in $\pspace$ to check whether it is satisfiable over $[0,1]$. If it is \emph{unsatisfiable}, there is no solution to~\eqref{equ:initialsystemmult} that makes the values of $x_{w^n_{1,1}:\Prob(\alpha)}$'s coherent with a~measure and the values of $x_{w^n_{1,1}:\psi}$'s with the semantics of~$\KLukProdProb$. In this case, we close $\Bmc$. Otherwise, we delete~\eqref{equ:initialsystemmult}, mark the constraint $w^{n-1}_{1,1}:\Box\tau\leq j$ that produced $w^n_{1,1}$ as ‘safe’, and delete the term $w^{n-1}_{1,1}Rw^n_{1,1}$. We then pick the next constraint in $w^{n-1}_{1,1}$ that can produce a~new state and repeat the process. Once all constraints of the form $w^{n-1}_{1,1}:\Box\tau'\leq j'$ are marked safe, we guess $\mathrm{poly}\Bmc^\Imbb_{w^{n-1}_{1,1}}$ and consider its conjunction with $\Bmc^\Imbb_{w^{n-1}_{1,1}}$. Note, furthermore, that if all probabilistic atoms occurring in $w^{n-1}_{1,1}$ are in the scope of~$\Box$, it suffices to check that $\Bmc^\Imbb_{w^{n-1}_{1,1}}$ has solutions. We repeat the process until either $w^0_1$ is marked safe (in which case, $\Gamma\not\models^\fb_\Prob\chi$) or all branches of our tableau are closed (whence, $\Gamma\models^\fb_\Prob\chi$).

Finally, recall that the depth of the model is bounded from above by $\dmc[\Xi]+1$; moreover, each state~$w$ contains only $\Bmc^\Imbb_w$ and $\mathrm{poly}\Bmc^\Imbb_w$, whose sizes are polynomial w.r.t.\ $\lmc[\Xi]$. Thus, we need polynomial space to execute the procedure. The result follows.
\end{proof}

We finish the section with two brief observations. First, in Example~\ref{example:multiagentprobabilisticframes}, we considered multi-agent probabilistic frames. One can readily see that our language can be straightforwardly expanded to incorporate such frames by allowing probabilistic atoms of the form $\Prob_a(\alpha)$ with $a\in\Amsf$. The proof of $\pspace$-completeness of the resulting logic can also be adapted expectedly. Second, we were dealing with the complexity of reasoning in $\KLukProdProb$ over \emph{finitely-branching} frames. Note, however, that we do not need the restriction to finitely branching frames if we consider the fragment $\Lmc_\mathsf{add}$ of $\LProb$ that contains all rational constants but \emph{does not} contain $\bullet$ and~$\rightarrow_\Pi$. We use $\KLukAddProb$ to denote the $\Lmc_\mathsf{add}$-fragment of~$\KLukProdProb$.

\begin{restatable}{theorem}{KLukProbPSpace}\label{theorem:KLukProbPSpace}
Let $\Gamma\cup\{\chi\}\subseteq\Lmc_\mathsf{add}$ be finite. Then, it is $\pspace$-complete to decide whether $\Gamma\models_\Prob\chi$.
\end{restatable}
\begin{proof}[Proof sketch]
By Proposition~\ref{prop:rationalconstantselimination}, we can eliminate rational constants from $\Gamma\cup\{\chi\}$ without the use of $\bullet$ and $\rightarrow_\Pi$. Thus, by~\cite[Lem\-ma~4.6]{Vidal2022}, the resulting logic is \emph{complete w.r.t.\ witnessed models}. Hence, the $\TKLukProd$-rules are sound. Now we use the same procedure as in the proof of Theorem~\ref{theorem:KLulmultProbpspace}.
\end{proof}

\section{Conclusion}\label{sec:conclusion}

Building on the work of Hájek et al.~\cite{Hajek1998,HajekEtAl1995,HajekEtAl2000,HajekHarmancova1995}, this paper investigated many-valued logics for reasoning about probability in the presence of modal notions such as knowledge, belief and action. Our main technical result, extending the contributions of \cite{MajerSedlar2025}, is that the problem of deciding local consequence over finitely branching frames is $\pspace$-complete for a~rather expressive logic. Over arbitrary frames, it is $\pspace$-complete for a~specific fragment of the full logic.  

Several interesting research problems remain to be tackled in the future. First, we would like to extend our complexity result for $\KLukProdProb$ beyond finitely branching frames and the result for $\KLukAddProb$ to stronger fragments. Interestingly enough, even for some fragments weaker than $\KLukProdProb$, this would require solving some open problems in many-valued logic. For example, $\mathbf{KP}\Luk^{\Prob}$ (the fragment without~$\onehalf$ and~$\rightarrow_\Pi$) is $\pspace$-complete if first-order $\mathbf{P}\Luk$ \cite{HorcikCintula2004} is complete with respect to witnessed models. We would also like to extend our results to modal logics stronger than $\mathbf{K}$. Second, we would like to compare the expressivity of our framework with the corresponding fragment of \cite{FaginHalpern1994}, especially when it comes to expressing statements about lower and upper probabilities. 
%
\bibliographystyle{eptcs}
\bibliography{generic}
\end{document}

%% file: macros.tex
\newcommand{\onehalf}{\mathbf{\tfrac{1}{2}}}

\newcommand{\Prop}{\mathtt{Prop}}

\newcommand{\Prob}{\mathsf{Pr}}

\newcommand{\LProb}{\mathcal{L}_\Prob}

\newcommand{\Lmodal}{\mathcal{L}_\Box}
\newcommand{\LCPL}{\mathcal{L}_{\{\sim,\wedge\}}}

\newcommand{\Luk}{{\normalfont{\textbf{{\L}}}}}
\newcommand{\KL}{{\normalfont\textbf{K{\L}}}}
\newcommand{\KLuk}{\mathbf{K}\Luk}

\newcommand{\KLukProdProb}{\KL\bm{\Pi}\onehalf^\mathsf{Pr}}

\newcommand{\KLukProd}{\KL\bm{\Pi}\onehalf}

\newcommand{\fb}{\mathsf{fb}}

\newcommand{\KLukAddProb}{\KLuk^\Prob_\Qmbb}

\newcommand{\Fmbb}{\mathbb{F}}
\newcommand{\Imbb}{\mathbb{I}}
\newcommand{\Nmbb}{\mathbb{N}}

\newcommand{\Qmbb}{\mathbb{Q}}
\newcommand{\Rmbb}{\mathbb{R}}

\newcommand{\Amsf}{\mathsf{A}}

\newcommand{\Cmsf}{\mathsf{C}}

\newcommand{\Gmsf}{\mathsf{G}}

\newcommand{\Wmsf}{\mathsf{W}}

\newcommand{\Bmc}{\mathcal{B}}

\newcommand{\Fmc}{\mathcal{F}}

\newcommand{\Imc}{\mathcal{I}}

\newcommand{\Lmc}{{\mathcal{L}}}

\newcommand{\Nmc}{{\mathcal{N}}}
\newcommand{\Omc}{{\mathcal{O}}}
\newcommand{\Pmc}{{\mathcal{P}}}

\newcommand{\Smc}{{\mathcal{S}}}
\newcommand{\Tmc}{{\mathcal{T}}}

\newcommand{\dmc}{\mathcal{d}}

\newcommand{\lmc}{\mathcal{l}}

\newcommand{\Ffrak}{\mathfrak{F}}

\newcommand{\Mfrak}{{\mathfrak{M}}}

\newcommand{\pfrak}{\mathfrak{p}}
\newcommand{\qfrak}{\mathfrak{q}}

\newcommand{\Kmbf}{\mathbf{K}}
\newcommand{\Mmbf}{\mathbf{M}}

\newcommand{\qmbf}{\mathbf{q}}

\newcommand{\pspace}{\mathsf{PSPACE}}
\newcommand{\np}{\mathsf{NP}}
\newcommand{\conp}{\mathsf{coNP}}

\newcommand{\TKLukProd}{\Tmc\!(\KLukProd_\fb)}
